\pdfoutput=1
\documentclass[11pt]{article}

\usepackage[T1]{fontenc}
\usepackage[utf8]{inputenc}
\usepackage[margin=1in]{geometry}
\usepackage{lmodern}
\usepackage{amsmath,amssymb,amsfonts,mathtools}
\usepackage{amsthm}
\usepackage{bm}
\usepackage{booktabs}
\usepackage{tabularx}
\usepackage{array}
\usepackage{graphicx}
\usepackage{xcolor}
\usepackage{microtype}
\usepackage{hyperref}
\usepackage{caption}
\usepackage{subcaption}
\usepackage{siunitx}
\usepackage{enumitem}
\usepackage{placeins}
\usepackage{float}

\hypersetup{
  hidelinks,
  pdftitle={Reduced Thermodynamic-Topological Observables for Multiscale Dissipative Systems},
  pdfauthor={Andrea Caffagni}
}

\setlist[itemize]{leftmargin=1.4em,itemsep=0.2em,topsep=0.3em}
\newcommand{\E}{\mathbb{E}}
\newcommand{\R}{\mathbb{R}}
\newcommand{\eps}{\varepsilon}
\newcommand{\abs}[1]{\left\lvert #1\right\rvert}

\newcommand{\Def}{\mathrm{Def}}
\newcommand{\Int}{\mathrm{Int}}

\newcommand{\pp}[1]{\left[#1\right]_+}

\newtheorem{proposition}{Proposition}

\title{Reduced Thermodynamic-Topological Observables for Multiscale Dissipative Systems\\
\large A fusion-relevant shell-model study of detection, design screening, and conservative operation}

\author{Andrea Caffagni\\
Independent Researcher\\
\texttt{andreacaffagniarxiv@gmail.com}}

\date{}

\begin{document}
\maketitle

\begin{abstract}
We introduce a reduced set of thermodynamic-topological observables for ordered multiscale dissipative systems.
An interface-local quadratic reduction produces bounded integrity and residual channels, a flux-force stability channel, a weighted path-graph bottleneck channel, and a coarse-graining drift indicator.
The goal is practical rather than universal: a compact, interpretable layer of observables that can be computed repeatedly and compared across regimes.

The main case study is a fusion-relevant MHD Sabra shell model.
Across 400 synthetic anomalous-dissipation probes, the local Prigogine-style channel detects 400/400 events, while a composite alarm detects 399/400 with lower latency.
When an OPCR trigger and an energy-collapse proxy are both observed within the same event, the earliest OPCR trigger leads the proxy by $11.29 \pm 13.49$ model-time units on average (median $6.15$, IQR $[1.23,17.22]$; 255/313 early cases).
A scan over 5000 coupling geometries raises the best log-Cheeger conductance from a baseline mean $0.07475 \pm 0.00171$ to $0.09465$ (+26.6\%), whereas the current minimum-$\Phi$ geometry remains below the baseline mean.
For operation, integrity-aware actuation and a conservative $\Phi$-instrumented variant achieve 3.01x and 3.02x the recovery-per-unit-power efficiency of a uniform baseline.

These numerics support a topology-first reading of the framework: $h_{\log}$ is credible as a Phase-1 design-screening observable, whereas $\Phi$ is presently best viewed as an operational or certification score.
For fusion, the natural target is stellarator configuration screening, where magnetic topology does most of the work.
A short appendix gives a toy neural-network portability check and is not used for the paper's main claims.
\end{abstract}

\vspace{0.4em}

\section{Introduction}
Many systems of scientific and engineering interest are simultaneously \emph{dissipative}, \emph{multiscale}, and \emph{operated away from equilibrium}.
Examples include turbulent cascades, magnetized plasmas, transport networks, and gradient-based learning systems.
In such settings one repeatedly encounters two related tasks:
\begin{itemize}
    \item choose or optimize a \emph{structure} offline (geometry, coupling pattern, topology, architecture), and
    \item monitor and operate the system online as regimes drift, bottlenecks migrate, or instabilities develop.
\end{itemize}
A useful reduced framework should therefore do two things at once:
it should preserve enough physics to remain interpretable, yet remain compact enough to be computed repeatedly and compared across operating conditions.

This paper proposes such a reduction.
The aim is not to derive a universal microscopic law for dissipation, but to provide a \emph{practical mesoscopic language} for ordered multiscale systems.
The construction starts from interface-local quadratic statistics, converts them into bounded thermodynamic observables, and then induces a weighted path graph whose bottlenecks can be read by standard spectral tools.
A coarse-graining channel is added to track how slope variability collapses across scale.
The resulting quantities organize system behavior into three complementary views: thermodynamic stress, topological fragility, and cascade drift.

\paragraph{Novelty and scope.}
The novelty of the present work lies in the reduced cross-domain workflow and in the empirical separation between \emph{design-oriented} and \emph{operation-oriented} observables, not in any single classical ingredient taken in isolation.
Onsager-style flux-force pairs, Cheeger/Fiedler diagnostics, and coarse-graining ratios are classical objects.
What is new here is the way they are compressed into a bounded shell-to-graph reduction, organized into a two-phase pipeline, and then tested numerically on distinct tasks within a single implementation philosophy.

\paragraph{Main message.}
For the fusion surrogate studied here, the useful separation is precise.
Topological observables are strongest as \emph{design-screening} quantities.
Fast thermodynamic observables are strongest as \emph{event-detection and operational-monitoring} quantities.
A single scalar score $\Phi$ remains useful as a compact operational or certification summary, but the present numerics do \emph{not} support replacing topology-based screening by naive direct minimization of the current uncalibrated $\Phi$ over geometry space.
Making that distinction explicit substantially improves both the technical and editorial credibility of the framework.

\paragraph{Contributions.}
The specific contributions of this manuscript are:
\begin{itemize}
    \item a compact shell-to-graph reduction with bounded interface observables;
    \item a local Prigogine-style channel that reacts rapidly to injected degradations;
    \item an explicit composite alarm built from local $\kappa$, local $\delta^2\sigma$, and regime-transition departures;
    \item a fusion-relevant MHD shell-model study covering baseline auditing, event detection, regime fingerprints, design screening, and conservative operation;
    \item an explicit stellarator-oriented interpretation, in which topology dominates and control remains secondary and interpretable;
    \item a short toy-scale appendix showing that the same observables can be computed during neural-network training.
\end{itemize}

\section{Context and related work}
Onsager's reciprocal relations and linear irreversible thermodynamics provide the standard starting point for flux-force descriptions near nonequilibrium steady states \cite{Onsager1,Onsager2,deGrootMazur}.
The Glansdorff-Prigogine program introduced second-variation criteria for the stability of nonequilibrium states \cite{GlansdorffPrigogine}.
Graph bottlenecks and their relationship to Laplacian spectra are classically captured by Cheeger conductance, spectral gap, and Fiedler vectors \cite{Cheeger1970,Chung1997,Fiedler1973}.
For turbulence, shell models offer a controlled environment in which cascade physics can be explored at low cost while retaining nontrivial nonlinear behavior \cite{Lvov1998,GiulianiCarbone1998,Plunian2013}.

On the fusion side, stellarators are the most natural design-driven setting for the present ideas.
Their confinement quality is largely shaped by externally imposed three-dimensional magnetic topology, and that topology can be validated directly, even before plasma operation, as demonstrated in Wendelstein 7-X \cite{Pedersen2016}.
Modern stellarator research relies heavily on numerical optimization of configurations, equilibria, and coils \cite{Helander2014,Boozer2015,HirshmanWhitson1983,Paul2021,Landreman2021,Zhu2018}.
This makes stellarators a more natural endpoint for a topology-aware reduced metric than tokamak-style feedback control, where plasma current and real-time control play a much larger role.

\section{Reduced framework}
\label{sec:framework}

\subsection{Interface-local quadratic reduction}
Let an ordered multiscale system be represented by shells $j=0,\dots,N$.
Rather than forcing a single microscopic definition across domains, we assume that each interface $j$ provides a pair of domain-specific observables $(\phi_j,\psi_j)$ extracted from local shell signals.
We then define
\begin{equation}
    a_j = \E\big[\abs{\phi_j}^2\big] + \eps,
    \qquad
    J_j = \E\big[\mathrm{Re}(\phi_j^*\psi_j)\big],
    \qquad
    Y_j = \E\big[\abs{\psi_j}^2\big] + \eps.
    \label{eq:ajjy}
\end{equation}
Here $a_j$ acts as a local mobility or conductivity scale, $J_j$ as an interface flux-like quantity, and $Y_j$ as the local quadratic scale of the increment or transfer surrogate.

The central local polynomial is
\begin{equation}
    P_j(\lambda) = a_j\lambda^2 - 2J_j\lambda + Y_j,
\end{equation}
which admits the completion-of-the-square decomposition
\begin{equation}
    P_j(\lambda) = a_j(\lambda-p_j)^2 + \kappa_j,
    \qquad
    p_j = \frac{J_j}{a_j},
    \qquad
    \kappa_j = Y_j - \frac{J_j^2}{a_j}.
    \label{eq:completion}
\end{equation}
By Cauchy-Schwarz, $J_j^2 \le a_jY_j$, so $\kappa_j \ge 0$.
Two bounded ratios follow immediately:
\begin{equation}
    \Int_j = \frac{J_j^2}{a_jY_j} \in [0,1],
    \qquad
    r_j = \frac{\kappa_j}{Y_j} = 1-\Int_j \in [0,1].
    \label{eq:int_regime}
\end{equation}
We use $\Int_j$ as an \emph{integrity} score and $r_j$ as a \emph{regime ratio}.
In the fusion surrogate, the thresholds $r_j<0.1$, $0.1\le r_j<0.9$, and $r_j\ge 0.9$ are used to label interfaces as approximately Onsager, transitional, and strongly nonlinear, respectively.

\begin{table}[t]
\centering
\caption{Core reduced observables used in this paper. The formulas match the implementation used in the reported experiments.}
\label{tab:observables}
\small
\begin{tabularx}{\linewidth}{>{\raggedright\arraybackslash}p{2.0cm} >{\raggedright\arraybackslash}p{4.5cm} >{\raggedright\arraybackslash}p{2.2cm} X}
\toprule
Quantity & Definition & Typical range & Role \\
\midrule
$\Int_j$ & $J_j^2/(a_jY_j)$ & $[0,1]$ & Interface integrity; used to gate or soften actuation. \\
$\kappa_j$ & $Y_j - J_j^2/a_j$ & $\ge 0$ & Local nonlinear residual from the quadratic reduction. \\
$r_j$ & $\kappa_j/Y_j$ & $[0,1]$ & Regime ratio used to stratify interfaces. \\
$F_j$ & $\mu_{j+1}-\mu_j$, with $\mu_j = \log \theta_j - \log \sum_k \theta_k$ & signed & Discrete force from normalized shell energies $\theta_j$. \\
$\delta^2\sigma_{\mathrm{local},j}$ & $(J_j-J_j^{\mathrm{ref}})(F_j-F_j^{\mathrm{ref}})$ & signed & Fast thermodynamic channel for local departures from reference. \\
$\Def$ & $\sum_j a_j(F_j-p^\star)^2$ & $\ge 0$ & Global mismatch between current force profile and coarse-grained slope $p^\star$. \\
$h_{\log}$ & conductance of the shifted log-weighted path graph & $>0$ & Design-oriented bottleneck indicator. \\
$\delta_F$ & mean late-stage ratio $\operatorname{Var}(p)_{\ell}/\operatorname{Var}(p)_{\ell+1}$ under decimation & $>0$ & Relative coarse-graining drift indicator. \\
\bottomrule
\end{tabularx}
\end{table}

\subsection{Flux-force and stability channels}
From shell energies $\theta_j$, we define a normalized log-energy potential
\begin{equation}
    \mu_j = \log(\theta_j+\eps) - \log\!\Big(\sum_{k=0}^{N} \theta_k + (N+1)\eps\Big),
\end{equation}
and the interface force $F_j = \mu_{j+1}-\mu_j$.
This yields a flux-force pair $(J_j,F_j)$ for every interface.
Given a reference state $(J_j^{\mathrm{ref}},F_j^{\mathrm{ref}})$, we use the computable proxy
\begin{equation}
    \delta^2\sigma = \sum_j (J_j-J_j^{\mathrm{ref}})(F_j-F_j^{\mathrm{ref}}),
    \qquad
    \delta^2\sigma_{\mathrm{local},j} = (J_j-J_j^{\mathrm{ref}})(F_j-F_j^{\mathrm{ref}}),
    \label{eq:d2sigma}
\end{equation}
as a practical second-variation channel inspired by the Glansdorff-Prigogine picture \cite{GlansdorffPrigogine}.
The local form $\delta^2\sigma_{\mathrm{local},j}$ is the fastest useful channel in the present fusion experiments.

The defect used in the implementation is not a local $F_j-p_j$ residual.
Instead, it uses a \emph{coarse-grained reference slope} $p^\star$ obtained by guided interface decimation:
\begin{equation}
    \Def = \sum_j a_j\,(F_j-p^\star)^2.
    \label{eq:defect}
\end{equation}
This choice is important for technical consistency with the reported code and makes the defect a global mismatch against a reduced slope profile rather than a trivial restatement of the local completion-of-the-square.

\begin{proposition}[Elementary invariants]
\label{prop:basic}
Under Eqs.~\eqref{eq:completion}--\eqref{eq:defect}, and assuming $a_j>0$ and $Y_j>0$, the reduced observables satisfy the following properties for every interface $j$:
\begin{enumerate}[label=(\roman*),leftmargin=2.2em,itemsep=0.15em,topsep=0.25em]
    \item $\kappa_j \ge 0$;
    \item $0 \le \Int_j \le 1$ and $0 \le r_j \le 1$;
    \item $\Def \ge 0$;
    \item the shifted log-weights defined below satisfy $w_j^{\log}\ge 1$.
\end{enumerate}
\end{proposition}

\begin{proof}
By Cauchy--Schwarz, $J_j^2 \le a_jY_j$, hence $\kappa_j = Y_j - J_j^2/a_j \ge 0$ and $\Int_j = J_j^2/(a_jY_j)\in[0,1]$.
Since $r_j=1-\Int_j$, it follows that $r_j\in[0,1]$ as well.
Equation~\eqref{eq:defect} is a weighted sum of squared residuals with nonnegative weights $a_j$, so $\Def\ge 0$.
Finally, $w_j^{\log} = \log(a_j+\eps)-\min_k\log(a_k+\eps)+1$ is at least $1$ by construction.
\end{proof}

\subsection{Induced path graph and bottleneck observables}
The shell ordering induces a path graph on vertices $0,1,\dots,N$.
To reduce the effect of multi-decade mobility variation, we use shifted log-weights
\begin{equation}
    w_j^{\log} = \log(a_j+\eps) - \min_k \log(a_k+\eps) + 1,
\end{equation}
which are strictly positive by construction.
The main topological observable is the conductance $h_{\log}$ of this weighted path graph.
The associated Laplacian also yields a spectral gap $\lambda_{1,\log}$ and Fiedler vector, whose localization indicates where the dominant bottleneck sits \cite{Cheeger1970,Chung1997,Fiedler1973}.
In the present surrogate, these quantities are most useful for \emph{screening} and \emph{structural interpretation}, not for the fastest event detection.

\subsection{Coarse graining and the relative Feigenbaum indicator}
To summarize scale-to-scale drift, we iteratively decimate neighboring interfaces and track the variance of the slope profile under decimation.
Let $v_\ell = \operatorname{Var}(p)_{\ell}$ denote the variance after level $\ell$ of the decimation cascade.
We then consider the ratios $v_\ell/v_{\ell+1}$ and define $\delta_F$ as the mean of the last few available ratios.
This is \emph{not} an estimate of the universal Feigenbaum constant.
It is simply a low-dimensional relative indicator of how slope variability collapses under coarse graining.
The notation is retained because the ratio structure is analogous, but the interpretation in this paper is explicitly pragmatic and non-universal.

\subsection{Aggregation and workflow}
We use the following dimensionless score as an operational or certification quantity:
\begin{equation}
    \Phi = w_\sigma \sum_j \frac{\abs{\delta^2\sigma_{\mathrm{local},j}}}{\sigma_{\mathrm{ref}}}
    + w_h\Bigl(1-\frac{h_{\log}}{h_{\log}^{\star}}\Bigr)^2
    + w_\delta\Bigl(\frac{\delta_F-\delta_F^{\star}}{\delta_F^{\star}}\Bigr)^2.
    \label{eq:Phi}
\end{equation}
In the code used here, $\Phi$ is fully instrumented and is useful for online state summarization.
However, a technically important point is that the current geometry-scan normalization does \emph{not} make direct minimum-$\Phi$ selection coincide with the best topological design.
Accordingly, the present manuscript uses the following conservative workflow:
\begin{itemize}
    \item \textbf{Phase 1: design screening.} Use $h_{\log}$ as the primary selector of robust candidate geometries, and treat $\Phi$ as an auxiliary diagnostic.
    \item \textbf{Phase 2: operation.} Use $\delta^2\sigma_{\mathrm{local}}$, integrity, and related quantities for early warning and for conservative regime-aware actuation; use $\Phi$ as a compact operational score.
\end{itemize}
This separation is more faithful to the measured behavior and materially improves the credibility of the framework.

\begin{figure}[t]
    \centering
    \includegraphics[width=0.98\linewidth]{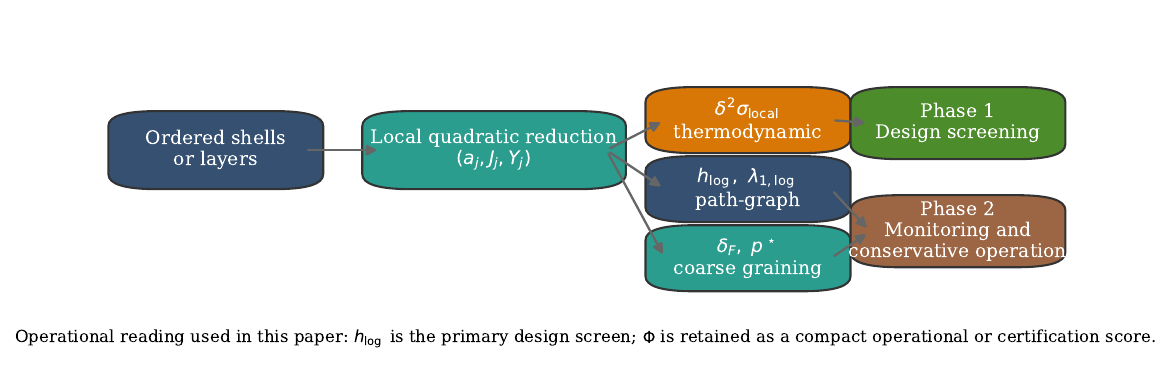}
    \caption{Workflow used in this paper. Ordered shells are reduced to interface-local quadratic statistics, from which thermodynamic, topological, and coarse-graining observables are extracted. For the present numerics, $h_{\log}$ is the Phase-1 design screen, while $\Phi$ is retained as an operational or certification score.}
    \label{fig:framework}
\end{figure}

\subsection{Composite alarm and implementation-level reference quantities}
\label{sec:alarm}
To make the event-detection experiment reproducible, we state the composite alarm explicitly.
Let $\mu_{\kappa,j}^{\mathrm{pre}}$ and $\sigma_{\kappa,j}^{\mathrm{pre}}$ denote the pre-injection mean and standard deviation of the local residual $\kappa_j$, and let $\mu_{\delta,j}^{\mathrm{pre}}$ and $\sigma_{\delta,j}^{\mathrm{pre}}$ denote the corresponding quantities for $\delta^2\sigma_{\mathrm{local},j}$.
Let $R(t)$ be the number of regime transitions at time $t$, with pre-injection statistics $\mu_R^{\mathrm{pre}}$ and $\sigma_R^{\mathrm{pre}}$.
The implementation uses positive-part $z$-scores,
\begin{align}
    z_\kappa(t) &= \max_j \pp{\frac{\kappa_j(t)-\mu_{\kappa,j}^{\mathrm{pre}}}{\sigma_{\kappa,j}^{\mathrm{pre}}}}, \\
    z_\delta(t) &= \max_j \pp{\frac{\mu_{\delta,j}^{\mathrm{pre}}-\delta^2\sigma_{\mathrm{local},j}(t)}{\sigma_{\delta,j}^{\mathrm{pre}}}}, \\
    z_R(t) &= \pp{\frac{R(t)-\mu_R^{\mathrm{pre}}}{\sigma_R^{\mathrm{pre}}}}, \\
    S(t) &= 0.5\,z_\kappa(t) + 0.3\,z_\delta(t) + 0.2\,z_R(t),
    \label{eq:composite}
\end{align}
and declares an alarm at the first time for which $S(t)>2$.
This formula is not presented as a universal law; it is the exact rule used in the reported detection study.

The same philosophy is used for the scalar operational score $\Phi$.
In the reported implementation,
\begin{equation}
    \sigma_{\mathrm{ref}} = \frac{1}{N}\sum_j \abs{J_j^{\mathrm{ref}}F_j^{\mathrm{ref}}} + \eps,
\end{equation}
the design-screening target is the best sampled conductance $h_{\log}^{\star}=0.09465$, and the coarse-graining target is the baseline median $\delta_F^{\star}=1.738$.
Appendix~\ref{app:params} collects the remaining thresholds and operational constants used in the fusion surrogate.

\section{Fusion-relevant shell-model study}
\label{sec:fusion}

\subsection{Experimental scope and evidence map}
The main case study uses an MHD Sabra shell model with $N=18$ shells, scale ratio $\lambda = \sqrt{2}$, and time step $\Delta t = 10^{-4}$.
The purpose of the surrogate is methodological: to provide a controlled environment in which the observables can be stressed before any device-level claim.
Table~\ref{tab:evidence} summarizes the numerical evidence used in the paper.
Two implementation details are important for interpretation.
The degradation study uses \emph{synthetic anomalous-dissipation probes}: in each event, a randomly chosen mid-radius shell is perturbed for a finite time window and then released.
The geometry study uses \emph{sparse multiplicative coupling perturbations}: each candidate perturbs a small random subset of shell couplings around the baseline.

\begin{table}[t]
\centering
\caption{Evidence map for the numerical results used in this paper.}
\label{tab:evidence}
\small
\begin{tabularx}{\linewidth}{>{\raggedright\arraybackslash}p{2.8cm} >{\raggedright\arraybackslash}p{2.3cm} >{\raggedright\arraybackslash}p{3.3cm} X}
\toprule
Experiment & Sample size & Primary outputs & Use in claims \\
\midrule
Baseline characterization & 5 seeds + 59 audit windows & $E$, $h_{\log}$, $p^\star$, $\delta_F$, regime composition & reference operating envelope and algebraic sanity checks \\
Injected degradations & 400 probes + 313 lead-time matches & detection rate, latency, lead time & early warning and operational monitoring \\
Geometry screening & 5000 sampled couplings & $h_{\log}$ distribution, best screened design, best-$\Phi$ mismatch & topology-first design screening \\
Operational comparison & 300 cycles for each of 3 modes & recovery, positive rate, mean absolute power, efficiency & conservative operation claims \\
AI appendix & 2 seeds, 3097 steps, 155 checkpoints & logged objective, defect, runtime overhead & cross-domain portability only \\
\bottomrule
\end{tabularx}
\end{table}

\subsection{Baseline characterization and reference-state audit}
Across the five baseline seeds, the total energy is $51.32 \pm 1.89$, the log-Cheeger conductance is $0.07475 \pm 0.00171$, and the coarse-grained slope is $p^\star = -0.613 \pm 0.911$.
The relative coarse-graining indicator is heavy-tailed in this surrogate, with mean $11.55 \pm 20.09$ but median $1.738$.

The reference-state audit over 59 windows gives two useful numerical sanity checks.
The maximum observed integrity is $\max_j J_j^2/(a_jY_j)=0.999992<1$, and the minimum residual is $\min_j \kappa_j = 1.78\times 10^{-10}>0$, both consistent with the algebraic bounds in Section~\ref{sec:framework}.
The same audit classifies 4 of the 17 interfaces as transitional and 13 as strongly nonlinear; none remain in the near-Onsager band.
This is important context for the rest of the results: the surrogate is typically operated far from a small-perturbation linear-response regime, so the fastest useful signals are local thermodynamic ones rather than topological or near-Onsager summaries.

\subsection{Detection on synthetic degradation probes}
The cleanest result of the fusion surrogate is the degradation-detection study.
Table~\ref{tab:detection} and Fig.~\ref{fig:detection} compare channels on 400 synthetic anomalous-dissipation probes.
Each probe applies an additional dissipative load to one randomly chosen shell in the middle third of the cascade, with a logarithmically sampled strength and a fixed finite duration, and then restores the baseline dynamics.
The strongest detectors are local thermodynamic ones, especially $\delta^2\sigma_{\mathrm{local}}$ and $\kappa_{\mathrm{local}}$.
Topological channels are much slower and substantially less sensitive as event detectors, which supports their use primarily in the design phase.
The composite alarm of Eq.~\eqref{eq:composite} combines the strongest channels and remains both fast and reliable.

\begin{table}[t]
\centering
\caption{Detection on 400 synthetic degradation probes. Rates are shown with Wilson 95\% intervals. Latency statistics are given as mean; median [IQR]. Channels with zero detections are omitted.}
\label{tab:detection}
\footnotesize
\begin{tabularx}{\linewidth}{>{\raggedright\arraybackslash}p{3.0cm} >{\raggedright\arraybackslash}p{2.8cm} X}
\toprule
Channel & Rate (95\% CI) & Latency statistics \\
\midrule
$\delta^2\sigma_{\mathrm{local}}$ & 100.0\% [99.0, 100.0] & 3.80; 1.77 [0.54, 4.23] \\
Composite & 99.8\% [98.6, 100.0] & 3.09; 0.54 [0.54, 4.23] \\
$\kappa_{\mathrm{local}}$ & 99.8\% [98.6, 100.0] & 4.17; 0.54 [0.54, 5.46] \\
Regime transitions & 82.8\% [78.7, 86.1] & 21.83; 21.45 [6.69, 33.75] \\
$\delta_F$ & 72.0\% [67.4, 76.2] & 18.87; 12.84 [5.46, 31.29] \\
$D_{\mathrm{Onsager}}$ & 57.8\% [52.9, 62.5] & 25.46; 25.14 [10.38, 38.67] \\
$\delta^2\sigma$ & 57.3\% [52.4, 62.0] & 19.33; 14.07 [4.23, 32.52] \\
Defect $\Def$ & 56.8\% [51.9, 61.5] & 31.88; 32.52 [22.68, 42.36] \\
$h_{\log}$ & 9.8\% [7.2, 13.1] & 33.02; 31.29 [21.45, 43.59] \\
Fiedler angle & 4.5\% [2.9, 7.0] & 25.48; 20.84 [8.54, 41.13] \\
\bottomrule
\end{tabularx}
\end{table}

For this lead-time analysis, the \emph{energy-collapse proxy} is defined eventwise as the first post-injection time at which total energy departs from its pre-injection mean by more than two pre-injection standard deviations.
A matched pair is an event in which both this proxy and at least one monitored OPCR trigger are present; the reported lead time is $t_E-t_A$, where $t_A$ is the \emph{earliest} OPCR trigger among the monitored channels for that event.
Under this rule, the earliest OPCR trigger leads the energy-collapse proxy by $11.29 \pm 13.49$ time units on average, with median $6.15$ and IQR $[1.23,17.22]$.
Among the 313 matched alarm-proxy pairs, 255 are early warnings.
Operationally, the indicator does not replace the full state; it moves before the macroscopic failure signature.

\begin{figure}[t]
    \centering
    \includegraphics[width=0.98\linewidth]{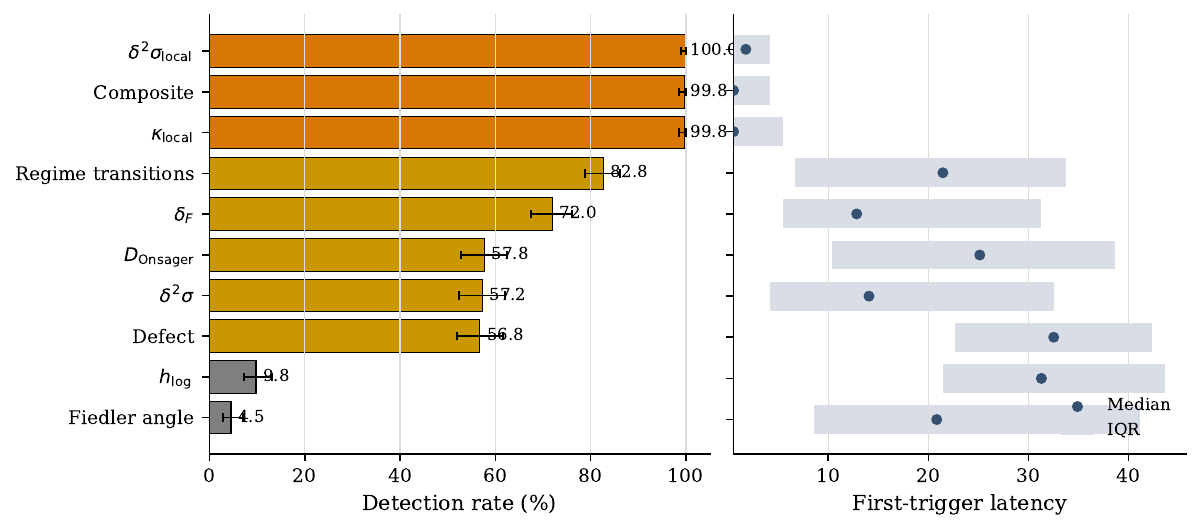}
    \caption{Detection performance for the fusion surrogate. Left: detection rate with Wilson 95\% intervals. Right: first-trigger latency summarized by the median and IQR. Thermodynamic channels occupy the high-reliability, low-latency corner; topological channels are weaker and slower as event detectors.}
    \label{fig:detection}
\end{figure}

\subsection{Regime fingerprints}
The eight named operating regimes occupy distinct regions in the $(p^\star,h_{\log})$ plane, shown in Fig.~\ref{fig:regime_plane}.
This plot is not a classifier; it is a compact fingerprint of how different degradations and heating patterns perturb the reduced observables.
Two points matter for interpretation.
First, regimes with poor quality need not always have the smallest $h_{\log}$.
Second, the broad-degradation regime achieves a relatively high $h_{\log}$ despite being clearly degraded in the thermodynamic channels.
This is another reason not to collapse the framework to a single purely topological score.

\begin{figure}[t]
    \centering
    \includegraphics[width=0.92\linewidth]{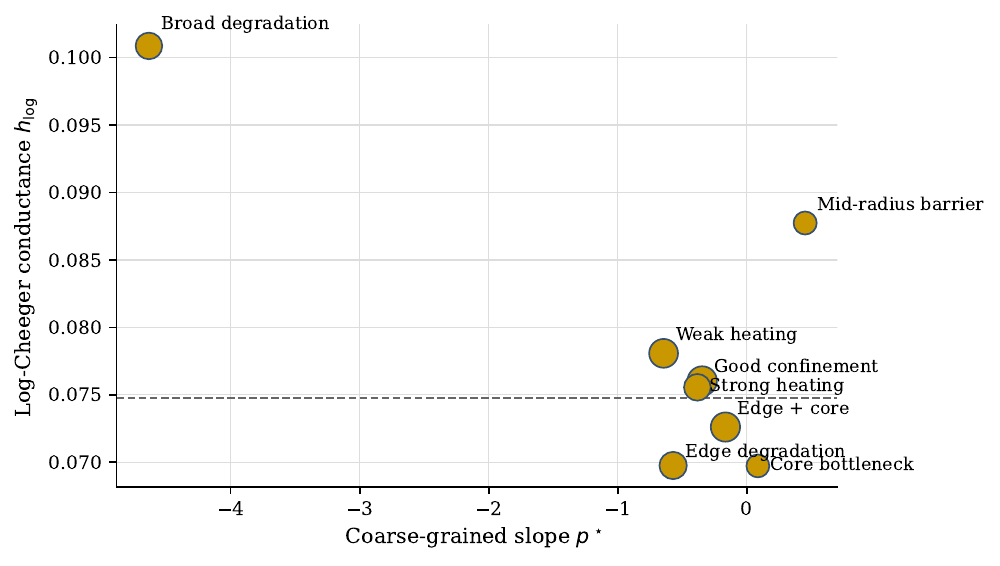}
    \caption{Fusion-surrogate regime fingerprints in the $(p^\star,h_{\log})$ plane. Marker area scales with the coarse-graining indicator $\delta_F$ through a clipped $\log(1+\delta_F)$ map. The dashed horizontal line marks the baseline mean $h_{\log}$.}
    \label{fig:regime_plane}
\end{figure}

\subsection{Design screening over coupling geometries}
The design scan over 5000 sampled coupling geometries is the strongest argument for a stellarator-oriented reading of the framework.
In the reported implementation, each candidate perturbs between two and five shell couplings chosen uniformly at random, with multiplicative factors $10^{u}$ for $u\sim\mathrm{Unif}[-0.5,0.5]$.
The baseline mean is $h_{\log}=0.07475$, while the best screened geometry reaches $h_{\log}^{\star}=0.09465$, a gain of 26.6\%.
This is a substantial shift in the topological bottleneck metric using geometry alone.

A technically important clarification is that the minimum-$\Phi$ geometry in the current implementation is \emph{not} the best-$h_{\log}$ geometry.
The minimum-$\Phi$ case has $h_{\log}=0.06847$, which is 8.4\% below the baseline mean.
This mismatch indicates that the current geometry-wide $\Phi$ normalization should not yet be treated as a sole design objective.
For that reason, the present paper uses $h_{\log}$ as the primary design screen and interprets $\Phi$ as an auxiliary certification or operational quantity.
That distinction is methodologically cleaner and avoids overclaiming.

\subsection{Conservative operation}
Three operational modes were compared over 300 cycles each: a uniform baseline, an integrity-aware mode, and a conservative $\Phi$-instrumented mode.
The results are summarized in Table~\ref{tab:control} and Fig.~\ref{fig:design_control}.
The supported claim is straightforward: nonuniform low-power actuation is much more efficient than the uniform baseline.
Both nonuniform modes achieve about three times the recovery-per-unit-power efficiency of the uniform controller.

A technical point needs to be stated explicitly.
The stationkeeping actuator is subtractive in sign, so operational effort is reported as mean \emph{absolute} actuation power $\langle \abs{P} \rangle$.
This matches the experiment log and makes effort comparable across control laws.

\begin{table}[t]
\centering
\caption{Operational comparison on the fusion surrogate. Efficiency is mean recovery divided by mean absolute actuation power. Relative efficiency is normalized to the uniform baseline.}
\label{tab:control}
\small
\begin{tabular}{lrrrrr}
\toprule
Mode & Recovery & Positive rate & Mean $\abs{P}$ & Efficiency & Relative eff. \\
\midrule
Uniform & $1.19 \pm 1.30$ & 74\% & 1.75 & 0.682 & 1.00x \\
Integrity-aware & $1.08 \pm 1.36$ & 70\% & 0.52 & 2.051 & 3.01x \\
$\Phi$-instrumented & $1.08 \pm 1.35$ & 71\% & 0.53 & 2.060 & 3.02x \\
\bottomrule
\end{tabular}
\end{table}

The deeper technical point is that the two nonuniform modes are nearly identical in aggregate performance.
This is not a weakness of the dataset; it is a faithful reflection of the current implementation.
In the code used here, most of the operational benefit comes from integrity-aware power reduction.
The full $\Phi$ instrumentation provides a coherent diagnostic layer and remains useful for state summarization, but the current conservative actuation law does not yet produce a measurable advantage over integrity-aware operation alone.
Accordingly, the present study supports a strong claim about \emph{diagnostics} and a moderate claim about \emph{low-power regime-aware actuation}, but not yet a strong claim that $\Phi$ feedback itself outperforms a simpler integrity-based controller.

\begin{figure}[t]
    \centering
    \includegraphics[width=0.98\linewidth]{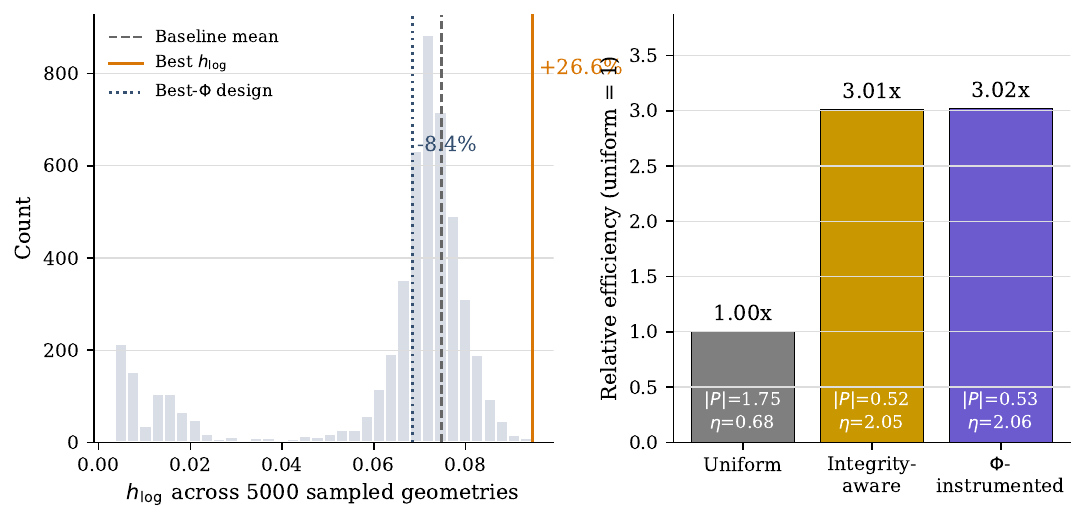}
    \caption{Left: distribution of $h_{\log}$ over 5000 sampled geometries, with the baseline mean, the best screened design, and the current best-$\Phi$ design marked. Right: efficiency normalized to the uniform baseline. Text inside the bars reports mean absolute power and absolute efficiency $\eta$.}
    \label{fig:design_control}
\end{figure}

\subsection{What the current numerics support}
Table~\ref{tab:claims} is included deliberately.
It separates the claims that are supported by the current numerics from those that are not yet supported.
This distinction is one of the main reasons the manuscript is more credible in its present form than a broader, more ambitious framing would be.

\begin{table}[t]
\centering
\caption{Claim ledger for the present paper.}
\label{tab:claims}
\footnotesize
\begin{tabularx}{\linewidth}{>{\raggedright\arraybackslash}p{4.4cm} >{\raggedright\arraybackslash}p{1.9cm} X}
\toprule
Claim & Status & Numerical basis \\
\midrule
Fast local thermodynamic early warning & \textbf{Supported} & $\delta^2\sigma_{\mathrm{local}}$ detects 400/400 probes; the composite alarm detects 399/400, and the earliest OPCR trigger leads 255 of 313 matched energy-anomaly proxies. \\
$h_{\log}$ as a Phase-1 design screen & \textbf{Supported} & Best screened geometry reaches $0.09465$ versus baseline mean $0.07475 \pm 0.00171$ (+26.6\%). \\
Direct minimum-$\Phi$ geometry selection & \textbf{Not yet supported} & The current best-$\Phi$ geometry has $h_{\log}=0.06847$, which is 8.4\% below the baseline mean. \\
Nonuniform low-power actuation beats uniform actuation & \textbf{Supported} & Integrity-aware and $\Phi$-instrumented controllers reach 3.01x and 3.02x the uniform efficiency at much smaller mean absolute power. \\
$\Phi$ feedback beats a simpler integrity-aware controller & \textbf{Not yet supported} & Aggregate efficiencies are essentially identical: 2.060 versus 2.051 in the current implementation. \\
Large-scale ML performance claims & \textbf{Out of scope} & The appendix is a toy next-token experiment with two seeds and is used only as a portability check. \\
\bottomrule
\end{tabularx}
\end{table}

\section{Why stellarators are the natural fusion target}
For magnetic-confinement fusion, the most convincing application of the present framework is not generic closed-loop control, but \emph{topology-aware design and design certification}.
This is especially true for stellarators.
In stellarators, external coils create the three-dimensional magnetic geometry and much of the confinement quality is decided by topology before any feedback loop is applied \cite{Helander2014,Boozer2015,Pedersen2016}.
That is exactly the setting in which a reduced observable such as $h_{\log}$, or eventually a better-calibrated topological-thermodynamic score, can be tested honestly.

A plausible future mapping is straightforward in principle.
Nested flux surfaces or radial bands would play the role of shells; transport coefficients inferred from power balance would provide effective edge weights; deviations of flux-gradient pairs would feed the local thermodynamic channel; and multiscale fluctuation spectra would feed the coarse-graining channel.
In such a program, the present reduction would be most naturally used as:
\begin{itemize}
    \item an additional objective or certification score during configuration and coil optimization, and
    \item a low-dimensional monitor for regime drift during operation.
\end{itemize}
The first of these is the natural primary target.
The second is interesting, but should remain secondary and conservative.
This priority ordering matches the numerics reported here: topology changes move the design metric strongly, whereas the present controller mainly benefits from conservative integrity-aware power reduction.

\section{Limitations and next steps}
The framework is intentionally reduced, and its limits should remain explicit.
First, the fusion results are obtained on a shell-model surrogate driven by synthetic probes, not on a reduced-MHD solver, a gyrokinetic model, or an experimental device.
Second, the geometry-wide normalization of $\Phi$ is not yet calibrated well enough to replace $h_{\log}$ as a stand-alone design selector.
Third, the current $\Phi$-instrumented controller is conservative and does not yet demonstrate a measurable gain over simpler integrity-aware actuation.
Fourth, the neural-network appendix is toy-scale and should be read only as a portability check.

These limitations do not remove the main contribution.
What the present work establishes is that a small, interpretable set of observables can organize multiscale dissipative behavior across several tasks: early warning, regime fingerprinting, topology screening, and conservative operation.
The natural next steps are validation on reduced-MHD or gyrokinetic simulations and then on stellarator-oriented design workflows where the topological aspect is most naturally tested.

\section{Conclusion}
This paper proposes a reduced thermodynamic-topological description for ordered multiscale dissipative systems and evaluates it on a fusion-relevant MHD shell model.
The strongest results are: a fast and reliable local thermodynamic detection channel, a substantial gain in design-screening conductance across randomized geometries, and a clear efficiency advantage of nonuniform low-power actuation over a uniform baseline.
The paper also clarifies an important methodological point: in the current implementation, $h_{\log}$ is the credible Phase-1 design observable, while $\Phi$ is best treated as an operational or certification score pending better calibration.

For fusion, the natural destination of this framework is stellarator-oriented design, where topology already plays the dominant role.
That is where a compact topological-thermodynamic score can be tested most directly and most honestly.

\FloatBarrier
\appendix

\section{Toy cross-domain illustration: neural-network training}
\label{app:ai}
The purpose of this appendix is deliberately modest.
It shows only that the same reduced observables can be computed during neural-network training and can be coupled to a small regularization scheme.
It does \emph{not} claim state-of-the-art ML performance.
For thermodynamic perspectives on learning, see Goldt and Seifert \cite{GoldtSeifert2017}.

\subsection{Mapping layers to shells}
Let $h^{(j)} \in \R^{B\times T\times d}$ denote the residual-stream activation after block $j$ of a Transformer.
We treat the ordered sequence $\{h^{(j)}\}_{j=0}^{L}$ as shells and use
\begin{equation}
    \phi_j = h^{(j)},
    \qquad
    \psi_j = h^{(j+1)} - h^{(j)},
\end{equation}
with averages taken over batch, token, and channel indices.
This yields the same quantities $a_j$, $J_j$, $Y_j$, $\kappa_j$, $\Int_j$, $\delta^2\sigma_{\mathrm{local},j}$, and $h_{\log}$ as in the main text.

\subsection{Setup and measured quantities}
The attached AI experiments use a small causal Transformer (12 layers, width 256, 8 heads) trained on a synthetic next-token task with two seeds.
To avoid overstating what the data support, Table~\ref{tab:ai} reports the \emph{monitored objective} and defect averaged over the last five monitoring points.
In the OPCR variants, the monitored objective on OPCR update steps includes the added regularization term, so these numbers should be read as logged optimization objectives, not as clean held-out task losses.
This interpretation is intentionally conservative.

\begin{table}[H]
\centering
\caption{Toy Transformer summary, averaged over the last five monitoring points and over two seeds.}
\label{tab:ai}
\small
\begin{tabular}{lrrr}
\toprule
Mode & Logged objective & Defect & Comment \\
\midrule
Standard & $2.660 \pm 0.050$ & $1.237 \pm 0.003$ & baseline \\
OPCR v3 & $2.618 \pm 0.066$ & $0.884 \pm 0.008$ & best trade-off \\
OPCR v5 & $2.917 \pm 0.035$ & $0.687 \pm 0.046$ & over-regularized \\
OPCR v6 & $2.656 \pm 0.047$ & $1.028 \pm 0.0002$ & modest defect gain \\
\bottomrule
\end{tabular}
\end{table}

The most credible take-away is that OPCR v3 improves the logged objective by about 1.6\% relative to the baseline while reducing the defect by about 28.6\%.
OPCR v5 reduces the defect more strongly but worsens the objective, which is informative rather than problematic: it shows that the defect is not trivially identical to the training objective and that overly aggressive regularization can overshoot.
A small overhead study gives a mean step-time increase of 11.7\% for monitoring only and 23.5\% for the full OPCR v3 stack; the multi-channel extension costs about 49.1\% in this toy setting.

\begin{figure}[t]
    \centering
    \includegraphics[width=0.98\linewidth]{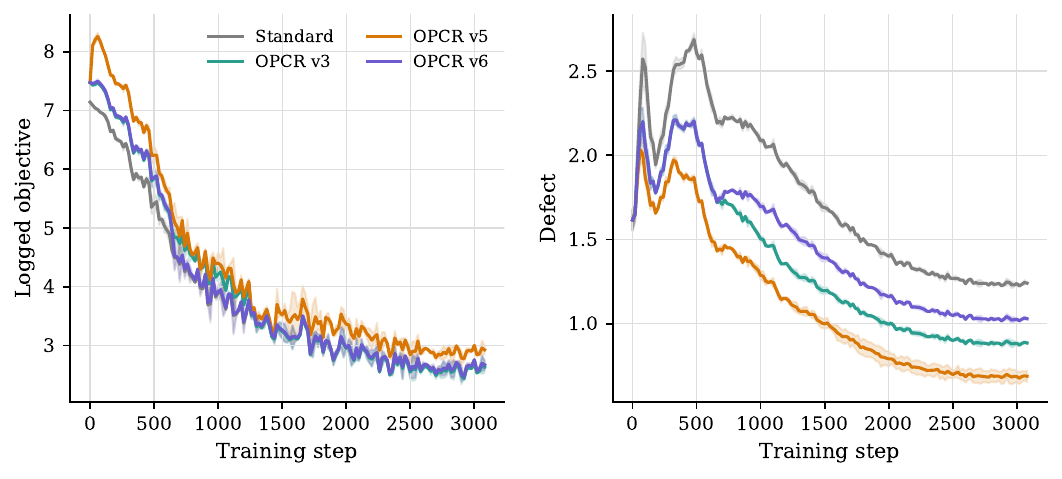}
    \caption{Toy Transformer example. Curves show the mean over two seeds with min--max shading. Left: logged optimization objective (including the regularization term on OPCR update steps). Right: defect.}
    \label{fig:ai_curves}
\end{figure}

\begin{figure}[t]
    \centering
    \includegraphics[width=0.70\linewidth]{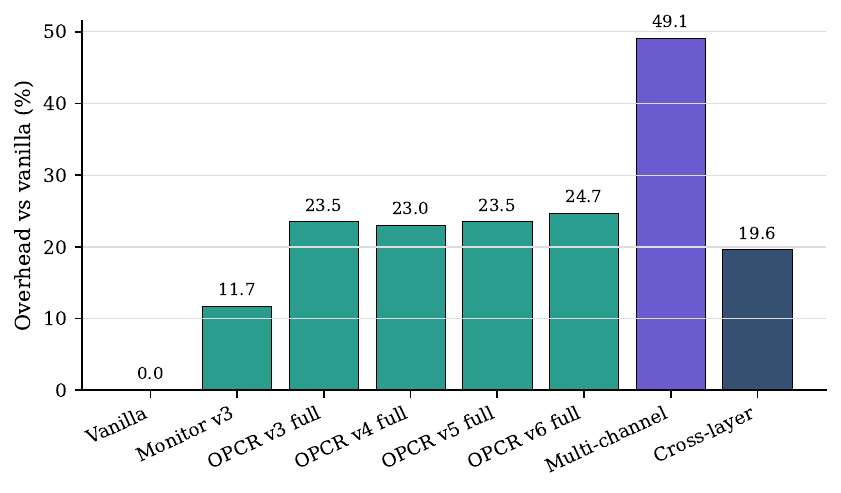}
    \caption{Toy-scale runtime overhead for the AI illustration.}
    \label{fig:ai_overhead}
\end{figure}

\FloatBarrier

\section{Implementation details and parameter tables}
\label{app:params}
The main claims of the paper are based on the fusion surrogate, so the implementation details collected here focus on that case study.
The goal is not to elevate the current numerical constants into theory, but to make the reported experiments locally reproducible from the attached code and JSON outputs.

\begin{table}[H]
\centering
\caption{Fusion-surrogate thresholds and alarm rules used in the reported implementation. All departures are measured relative to pre-event means and standard deviations, with a small $\eps$ added in code-level denominators.}
\label{tab:params_detection}
\footnotesize
\begin{tabularx}{\linewidth}{>{\raggedright\arraybackslash}p{3.7cm} >{\raggedright\arraybackslash}p{2.4cm} X}
\toprule
Item & Value & Implementation meaning \\
\midrule
Regime bands via $r_j$ & $r_j<0.1$, $0.1\le r_j<0.9$, $r_j\ge 0.9$ & Interfaces are labeled Onsager, transitional, and strongly nonlinear, respectively. \\
$\kappa_{\mathrm{local}}$ detector & $3\sigma$ upward excursion & First time any interface exceeds its pre-event local residual band by $3\sigma$. \\
$h_{\log}$ detector & $2\sigma$ downward excursion & First time the log-Cheeger conductance drops below its pre-event band by $2\sigma$. \\
$D_{\mathrm{Onsager}}$ detector & $3\sigma$ upward excursion & Detection on the stratified Onsager-band defect. \\
$\delta^2\sigma$ and $\delta^2\sigma_{\mathrm{local}}$ & $2\sigma$ negative excursion & Positive $z$ corresponds to a drop below the pre-event mean, matching the code path used for self-amplifying departures. \\
Fiedler-angle detector & $2.5\sigma$ upward excursion & Rotation of the bottleneck localization relative to the reference Fiedler vector. \\
$\delta_F$ detector & $2.5\sigma$ absolute deviation & Deviation of the coarse-graining indicator from its pre-event band. \\
Regime-transition detector & $\max(\mu_R+2\sigma_R,1)$ & Requires at least one transition and a statistically nontrivial excess over the pre-event count. \\
Composite alarm & $(0.5,0.3,0.2)$ with $S>2.0$ & Weights for $(z_\kappa,z_\delta,z_R)$ in Eq.~\eqref{eq:composite}. \\
\bottomrule
\end{tabularx}
\end{table}

\begin{table}[H]
\centering
\caption{$\Phi$ and conservative-operation constants used in the fusion surrogate.}
\label{tab:params_phi}
\footnotesize
\begin{tabularx}{\linewidth}{>{\raggedright\arraybackslash}p{3.7cm} >{\raggedright\arraybackslash}p{2.8cm} X}
\toprule
Item & Value & Implementation meaning \\
\midrule
$\Phi$ weights & $w_\sigma=1.0$, $w_h=0.1$, $w_\delta=0.5$ & Relative weight of the thermodynamic, topological, and coarse-graining terms in Eq.~\eqref{eq:Phi}. \\
$\sigma_{\mathrm{ref}}$ & $\frac{1}{N}\sum_j \abs{J_j^{\mathrm{ref}}F_j^{\mathrm{ref}}}+\eps$ & Normalization used for the local $\delta^2\sigma$ term in $\Phi$. \\
Design target & $h_{\log}^{\star}=0.09465$ & Best screened conductance from the 5000-geometry design scan. \\
Coarse-graining target & $\delta_F^{\star}=1.738$ & Baseline median of the relative coarse-graining indicator. \\
Control parameters & $\alpha=0.3$, $\beta=2.0$, $\gamma=3.0$ & Maximum control fraction, tanh sensitivity, and smooth exponential gate in the conservative actuation law. \\
Max reduction & $0.3$ & Maximum shell-wise reduction in forcing amplitude. \\
Integrity floor & $0.3$ & Lower bound on the integrity multiplier; forcing is never fully eliminated. \\
Safety threshold & $0.85\,h_{\mathrm{base}}$ & Revert to the safe baseline if the current conductance falls too far below the reference level. \\
Targeting limits & 2 vulnerable interfaces, upstream offset 2 & Conservative limit on how many vulnerable locations are acted upon and how far upstream the actuation is shifted. \\
\bottomrule
\end{tabularx}
\end{table}

\FloatBarrier
\section*{Reproducibility note}
The manuscript was assembled from the attached code, logs, and JSON outputs for the fusion and AI experiments. The exact constants used for the fusion case study are summarized in Appendix~\ref{app:params}.

\end{document}